\documentclass{entcs}
\usepackage[utf8]{inputenc}
\usepackage[T1]{fontenc}
\usepackage{entcsmacro}
\usepackage{graphicx}
\usepackage{subcaption}
\usepackage{tikz}
\usetikzlibrary{calc,decorations,decorations.pathreplacing,decorations.shapes,decorations.markings,shapes}

\tikzstyle{ivertex}=[circle,fill=black!25,minimum size=10pt,inner sep=0pt,draw=black!80]
\tikzstyle{wvertex}=[circle,fill=black!15,minimum size=8pt,inner sep=0pt,draw=black!80]
\tikzstyle{svertex}=[circle,fill=black!25,minimum size=7pt,inner sep=0pt,draw=black!80]
\tikzstyle{evertex}=[ellipse,fill=black!25,minimum width=10pt,minimum height=12pt,inner sep=0pt,draw=black!80]
\tikzstyle{bvertex}=[circle,fill=black!25,minimum size=20pt,inner sep=0pt,draw=black!80]
\tikzstyle{edge}=[draw,-]

\usepackage{algorithm}
\usepackage{algpseudocode}

\algdef{SE}[DOWHILE]{Do}{doWhile}{\algorithmicdo}[1]{\algorithmicwhile\ #1}%
\algblockdefx[MyBlock]{Begin}{End}%
{\textbf{begin}}%
{\textbf{end}}
\newcommand\MyFunction[2]{\State \textbf{function} \textsc{#1}(#2)}

\makeatletter
\newcommand{\settboxlength}{%
  \settowidth{\@tempdima}{\algorithmicensure}%
  \setlength{\@tempdima}{\dimexpr\linewidth-\@tempdima-1em+\@totalleftmargin}}
\newcommand{\tbox}[1]{%
  \settboxlength%
  \parbox[t]{\@tempdima}{\strut #1\strut}}
\makeatother

\usepackage{mathtools}
\newcommand{\ceil}[1]{\left\lceil#1\right\rceil}
\newcommand{\floor}[1]{\left\lfloor#1\right\rfloor}

\def\lastname{Gudmundsson, Magnusson, Saemundsson}

\begin{document}

\begin{frontmatter}

    \title{Bounds and Fixed-Parameter Algorithms for Weighted Improper Coloring\\ (Extended version)}

    \author{Bjarki Agust Gudmundsson \thanksref{bjarkimail}}
    \address{ICE-TCS, School of Computer Science\\Reykjav\'ik University\\ Reykjav\'ik, Iceland}

    \author{Tomas Ken Magnusson \thanksref{tomasmail}}
    \address{ICE-TCS, School of Computer Science\\Reykjav\'ik University\\ Reykjav\'ik, Iceland}

    \author{Bjorn Orri Saemundsson \thanksref{bjornmail}}
    \address{ICE-TCS, School of Computer Science\\Reykjav\'ik University\\ Reykjav\'ik, Iceland}

    \thanks[bjarkimail]{Email: \href{mailto:bjarkig12@ru.is} {\texttt{\normalshape bjarkig12@ru.is}}}
    \thanks[tomasmail]{Email: \href{mailto:tomasm12@ru.is} {\texttt{\normalshape tomasm12@ru.is}}}
    \thanks[bjornmail]{Email: \href{mailto:bjorn12@ru.is} {\texttt{\normalshape bjorn12@ru.is}}}

    \begin{abstract}
        We study the weighted improper coloring problem, a generalization of
        defective coloring. We present some hardness results and in particular
        we show that weighted improper coloring is not fixed-parameter
        tractable when parameterized by pathwidth. We generalize bounds for
        defective coloring to weighted improper coloring and give a bound for
        weighted improper coloring in terms of the sum of edge weights. Finally
        we give fixed-parameter algorithms for weighted improper coloring both
        when parameterized by treewidth and maximum degree and when
        parameterized by treewidth and precision of edge weights. In
        particular, we obtain a linear-time algorithm for weighted improper
        coloring of interval graphs of bounded degree.
    \end{abstract}

\end{frontmatter}

\clearpage

\section{Introduction}

Graph coloring is a classic subject of both mathematics and computer science
with many practical applications. It was one of Karp's $21$ original
$\mathcal{NP}$-complete problems. Many generalizations of graph coloring, such
as defective coloring, have been studied, but most of those apply to undirected
graphs. In this paper we consider weighted improper coloring, a generalization
of graph coloring for weighted digraphs. Weighted improper coloring has not
received much attention until recently in the context of wireless networks. In
some models for wireless networks, such as the SINR model, communication over
one wireless link can disturb communication over other wireless links. This
disturbance can vary from link to link, and depends on signal strength and
other variables of the surrounding environment.  Furthermore the disturbance
does not need to be symmetric. A scheduling of communications in such a network
can be modeled as a weighted improper coloring, and this is where the
generalization stems from.

In this paper we provide some new bounds, fixed-parameter algorithms and
hardness results for weighted improper coloring, some of which are
generalizations of existing results for defective coloring.

\subsection{Preliminaries}
Let $G=(V,E)$ be an undirected graph. For a vertex $v\in V$ let $d(v)$ denote
the degree of $v$ in $G$, and $d_S(v)$ denote the degree of $v$ in the subgraph
induced by a subset $S \subseteq V$. A \textit{$k$-coloring} $c:V\rightarrow
\{1,\ldots,k\}$ of $G$ is a partition of the vertex set $V$ into $k$
vertex-disjoint subsets, and $c[v]$ denotes the set of vertices that have color
$c(v)$. A $k$-coloring $c$ is \textit{$d$-defective} if for each $v\in V$,
$d_{c[v]}(v) \leq d$. A $d$-defective $k$-coloring is also called a
\textit{$(k,d)$-defective coloring}. Note that ordinary coloring is a special
case of defective coloring where $d=0$, so defective coloring is a proper
generalization of ordinary coloring.

Let $G=(V,E,w)$ be a weighted digraph, where $w:E\rightarrow [0,1]$ is an
associated weight function. For a vertex $v\in V$ let $d^{-}(v) =
\sum_{(u,v)\in E} w(u,v)$ denote its weighted indegree, and $d_S^{-}(v)$ denote
the weighted indegree of $v$ in the subgraph induced by a subset $S\subseteq
V$. Let $\Delta^{-}$ denote maximum of any weighted indegree.

\begin{definition}
    A $k$-coloring $c$ of a weighted digraph $G=(V,E,w)$ is a \textit{weighted
    improper $k$-coloring} if for each vertex $v\in V$, $d^{-}_{c[v]}(v) < 1$.
    The weighted improper chromatic number $\chi_w(G)$ is the minimum number
    $k$ such that $G$ has a weighted improper $k$-coloring with respect to the
    weight function $w$.
\end{definition}

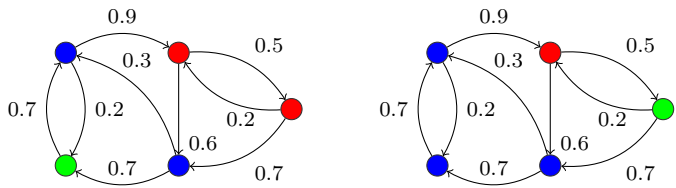
\begin{figure}[h]
  \centering
  \begin{tikzpicture}[auto]
    \node[wvertex,fill=green] (0) at (0,0) {};
    \node[wvertex,fill=blue] (1) at (1.5,0) {};
    \node[wvertex,fill=blue] (2) at (0,1.5) {};
    \node[wvertex,fill=red] (3) at (1.5,1.5) {};
    \node[wvertex,fill=red] (4) at (3,0.75) {};

    \path[->,every node/.style={font=\scriptsize}]
      (1) edge[bend left] node[swap] {$0.7$} (0)
      (0) edge[bend left] node {$0.7$} (2)
      (2) edge[bend left] node {$0.2$} (0)
      (1) edge[bend right] node[swap, near end, yshift=-3pt] {$0.3$} (2)
      (2) edge[bend left] node {$0.9$} (3)
      (3) edge node[very near end] {$0.6$} (1)
      (3) edge[bend left] node {$0.5$} (4)
      (4) edge[bend left] node[very near start, yshift=3pt] {$0.2$} (3)
      (4) edge[bend left] node {$0.7$} (1);
  \end{tikzpicture}
  \quad\quad
  \begin{tikzpicture}[auto]
    \node[wvertex,fill=blue] (0) at (0,0) {};
    \node[wvertex,fill=blue] (1) at (1.5,0) {};
    \node[wvertex,fill=blue] (2) at (0,1.5) {};
    \node[wvertex,fill=red] (3) at (1.5,1.5) {};
    \node[wvertex,fill=green] (4) at (3,0.75) {};

    \path[->,every node/.style={font=\scriptsize}]
      (1) edge[bend left] node[swap] {$0.7$} (0)
      (0) edge[bend left] node {$0.7$} (2)
      (2) edge[bend left] node {$0.2$} (0)
      (1) edge[bend right] node[swap, near end, yshift=-3pt] {$0.3$} (2)
      (2) edge[bend left] node {$0.9$} (3)
      (3) edge node[very near end] {$0.6$} (1)
      (3) edge[bend left] node {$0.5$} (4)
      (4) edge[bend left] node[very near start, yshift=3pt] {$0.2$} (3)
      (4) edge[bend left] node {$0.7$} (1);
  \end{tikzpicture}
  \label{fig:wimpexample}
  \caption{A valid weighted improper coloring of the graph to the left. The
  graph on the right has an invalid coloring as the upper left corner vertex
  has indegree $1$ from same-colored neighbors.}
\end{figure}

Now notice that weighted improper coloring is a generalization of defective
coloring, and by extension a generalization of ordinary coloring, as is
captured by the following lemma.

\begin{lemma}
    \label{thm:def_imp_reduction}
    Defective $(k,d)$-coloring can be reduced to weighted improper $k$-coloring in
    polynomial time.
\end{lemma}
\begin{proof}
  Let $G$ be an undirected graph. In a valid defective $(k,d)$-coloring of $G$
  any vertex $v \in V(G)$ can have at most $d$ adjacent vertices of the same
  color. Let $G'$ be the weighted digraph obtained by replacing every
  edge $\{u,v\} \in E(G)$ with two directed edges $(u,v)$ and $(v,u)$ with weight
  $\frac{1}{d + 1}$. Clearly, any vertex $v' \in V(G')$ cannot have more than
  $d$ adjacent vertices of the same color because then we would have
  $d_{c[v']}^-(v') \geq \frac{d + 1}{d + 1} = 1$. However, $v'$ can have up to
  $d$ adjacent vertices of the same color because that yields $d_{c[v']}^-(v')
  \leq \frac{d}{d + 1} < 1$. Thus, any valid weighted improper $k$-coloring of
  $G'$ is also a valid defective $(k,d)$-coloring of $G$ and any invalid
  weighted improper $k$-coloring of $G'$ is an invalid defective
  $(k,d)$-coloring of $G$. The reduction has time complexity $O(|E(G)|)$ and is
  thus a polynomial-time reduction.
\end{proof}

As this reduction does not change the underlying structure of the graph, but
only adds weights to it, we get the following corollary.

\begin{corollary}
    \label{cor:def_imp_reduction_spec}
    Defective $(k,d)$-coloring of any specific class of graphs has a polynomial time reduction to weighted
    improper $k$-coloring of the same class of graphs.
\end{corollary}

Finally we define a tree decomposition as follows:
\begin{definition}
    \label{def:treedecomp}

    A tree decomposition of a graph $G=(V,E)$ is a pair $(X,T)$, where
    $X=\{X_1,\ldots,X_n\}$ is a collection of subsets of $V$, and $T$ is a tree
    whose vertices are the subsets in $X$, which we will refer to as
    super-vertices. Additionally, the following properties must hold:
    \begin{enumerate}
        \item if $v \in V$, then there exists a subset $X_i$ such that $v\in X_i$, \label{prop:treedecompvertex}
        \item if $(u,v) \in E$, then there exists a subset $X_i$ such that $u\in X_i$ and $v\in X_i$, and \label{prop:treedecompedge}
        \item if $v\in V$ and $X^v \subseteq X$ is the set of subsets that
            contain $v$, then $X^v$ forms a connected subtree of $T$. \label{prop:treedecompsubtree}
    \end{enumerate}

    The width of a tree decomposition $(X,T)$ is $\max_i |X_i|-1$. The treewidth
    of a graph $G$ is the minimum width of any tree decomposition of $G$.
    A path decomposition is a tree decomposition $(X,T)$ where $T$ is a path
    graph. The pathwidth of a graph $G$ is the minimum width of any path
    decomposition of $G$.
    Finally, path and tree decompositions can be extended to digraphs by
    using the underlying graph (i.e.\ by treating directed edges as undirected).
\end{definition}

\textbf{Previous work}\hspace{1em}

Defective coloring and the related $t$-improper coloring were first introduced
by Andrews and Jacobsen \cite{andrews1985generalization}, Harary and Frank \cite{harary1985conditional} and Cowen et al.\ \cite{cowencowen}.
%
Graphs on embeddable surfaces were the main focus of Cowen et al.\
\cite{cowencowen} and Cowen, Goddard and Jeserum \cite{cowen1997defective} and
they
characterized all $(k,d)$ such that planar graphs are $d$-defective
$k$-colorable and produced results for surfaces of higher genus.
Frick and Henning \cite{frick1994extremal} gave extremal results on the
defective chromatic number and Kang and McDiarmid \cite{timpr-random} proved
bounds on the growth of the defective chromatic number of random graphs.  Other
properties than degrees of vertices have been researched regarding defective
coloring and Frick \cite{frick1993survey} gives a good survey on such variations
of the problem. 

The first to propose this edge weighted variation for undirected graphs were
Hoefer, Kesselheim and Vöcking \cite{hoefer2014approximation}. Araujo et
al.\ \cite{araujo} defined the problem with a variable threshold on the degree of
a vertex with respect to its coloring and explored the dual of the problem of
finding such a threshold, in addition to producing results on various grid
graphs.  Similar channel assignment problems have also been considered and
modeled as colorings on vertex weighted graphs \cite{bermond2010improper}.

Tamura et al.\ \cite{tamura} and Archetti et al.\ \cite{archetti2015a} present
a clear formulation of wireless scheduling of the SINR model systems as
directed weighted improper coloring and most of the work has assumed
geometric restrictions on the interference, while some general results have
been published. Recently Halldórsson and Bang-Jensen \cite{hallbang} gave the
essentially tight bound $\chi_w(G) \leq \floor{2\Delta^{-} + 1}$.

\section{Hardness}
Some hardness results about defective coloring have been established. Cowen
and Jeserum \cite{cowen1997defective} proved that $(2,d)$-defective coloring is
$\mathcal{NP}$-complete for $d\geq 1$, even for planar graphs, and that
$(3,1)$-defective coloring is also $\mathcal{NP}$-complete for planar graphs.
Furthermore they show that $(k,d)$-defective coloring is
$\mathcal{NP}$-complete for any $k\geq 3, d\geq 0$.

These results can be carried over to weighted improper coloring, as we show in
the following corollaries.

\begin{corollary}
    Weighted improper $k$-coloring is $\mathcal{NP}$-complete for $k \geq 2$.
  \label{cor:wimpnpcomplete}
\end{corollary}
\begin{proof}
    This follows from Lemma~\ref{thm:def_imp_reduction} and the fact that
    defective $(k,d)$-coloring is $\mathcal{NP}$-complete for $k=2$ and $d\geq 1$, as
    well as $k \geq 3$ and $d\geq 0$.
\end{proof}

Note that, as there is a simple polynomial-time algorithm for $2$-coloring, it
is a bit surprising that defective $2$-coloring and weighted improper
$2$-coloring are $\mathcal{NP}$-complete.

\begin{corollary}
    \label{cor:planar_np_complete}
    Weighted improper $k$-coloring is $\mathcal{NP}$-complete for planar graphs when $k\in\{2,3\}$.
\end{corollary}
\begin{proof}
    This follows from Corollary~\ref{cor:def_imp_reduction_spec} and the fact that
    $(2,d)$-defective coloring is $\mathcal{NP}$-complete for planar graphs when $d\geq 1$, as
    well as the fact that $(3,1)$-defective coloring is $\mathcal{NP}$-complete for planar graphs.
\end{proof}

Since any planar graph is $4$-colorable \cite{appel1980every}, which is a
stricter requirement than being weighted improper $4$-colorable, and as
$1$-colorability is trivial, this, along with
Corollary~\ref{cor:planar_np_complete}, gives complete hardness results for
weighted improper coloring of planar graphs.

Now consider a graph $G$, and say we add any amount of $0$-weight edges to this
graph to make a new graph $G'$. As these $0$-weight edges impose no new
restrictions, it is clear that $G$ is $k$-colorable if and only if $G'$ is
$k$-colorable. In fact, we can add $0$-weight edges until the graph is complete
without adding any restrictions, and hence, in a sense, we can always assume
that the graph is complete. This observation gives rise to the following lemma.

\begin{lemma}
  Let $\mathcal{G}$ be a family of graphs such that for each $n$ the
  complete graph $K_n \in \mathcal{G}$. Then weighted improper $k$-coloring for
  $\mathcal{G}$ is $\mathcal{NP}$-complete.
  \label{lem:cliquesnpcomplete}
\end{lemma}
\begin{proof}
    We show this by reducing general weighted improper $k$-coloring, which is
    $\mathcal{NP}$-complete by Corollary~\ref{cor:wimpnpcomplete}, to weighted
    improper $k$-coloring for $\mathcal{G}$. Take any weighted digraph $G =
    (V,E,w)$. Create a complete weighted digraph $G'$ from $G$ by adding
    $0$-weight edges between pairs of vertices that are not connected by an
    edge in $G$. Clearly $G$ is $k$-colorable if and only if $G'$ is
    $k$-colorable. Since $G'$ is a complete graph, we see that $G'\in
    \mathcal{G}$, which concludes the reduction.
\end{proof}

In particular, this lemma implies the $\mathcal{NP}$-completeness of weighted
improper coloring for interval graphs, $C_n \times K_t$, and $k$th powers of
paths and cycles, where $t$ and $k$ are unbounded. Later we will give fixed-parameter
algorithms for these and other graph classes.

This lemma can be further generalized by noticing that vertices with $0$-weight
edges can also be added to a graph without adding any restrictions, and hence
graphs with large enough complete subgraphs can be used instead of complete
graphs in the reduction.

\begin{theorem}
    \label{thm:pathwidthhardness}
    Weighted improper coloring is $\mathcal{NP}$-complete for graphs of bounded
    pathwidth.
\end{theorem}
\begin{proof}
  Given a multiset $S = \{x_1,\ldots,x_n\}$ of positive integers, the partition
  problem is the task of deciding whether $S$ can be partitioned into two sets
  $S_1,S_2$ such that $\sum_{x \in S_1}x = \sum_{x \in S_2}x$.  We show that
  weighted improper coloring is $\mathcal{NP}$-complete for graphs of bounded
  pathwidth by reducing the partition problem to weighted improper coloring of
  a graph with bounded pathwidth.

  We construct an undirected graph $G = (V,E)$ such that  for every integer
  $x_i \in S$ we add a corresponding vertex $v_i$ to $V$, along with two additional
  vertices $A$ and $B$. We add an undirected edge $\{A,B\}$ with weight $1$ and
  for each vertex $v_i$ we add two undirected edges $\{A, v_i\}$ and $\{v_i,
  B\}$ to the graph with weights $w_i = 2x_i/X - \epsilon/|S|$ where $X =
  \sum_{x\in S} x$. The graph is depicted in
  Figure~\ref{fig:subsetsumcoloring}.

  \begin{figure}[h]
    \centering
    \begin{tikzpicture}
      \node[ivertex] (a) {\scriptsize $A$};
      \node[ivertex] (b) at ($(a) - (0,3)$) {\scriptsize $B$};
      \path (a) edge node [left] {\scriptsize$1$} (b);
      \foreach \x[count=\xi] in {1,2,n} {
          \node[wvertex] (\x) at ($(a) + (\xi - 0.5, -1.5) $) {\scriptsize$v_\x$};
          \path (a) edge node [right,near end] {\scriptsize$w_\x$} (\x);
          \path (b) edge node [right,near end] {\scriptsize$w_\x$} (\x);
      }
      \path (2) -- (n) node [midway] {$\dots$};
    \end{tikzpicture}
    \caption{Subset sum modeled as a weighted improper coloring.}
    \label{fig:subsetsumcoloring}
  \end{figure}
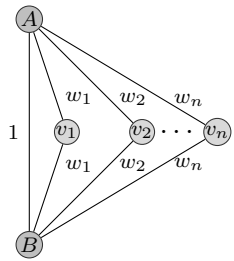

  We claim that $S$ can be partitioned into $S_1$ and $S_2$ if and only if
  there exists a valid weighted improper $2$-coloring of $G$.

  Given a partition of $S$ into $S_1,S_2$ we assign to every vertex $v_i$
  corresponding to an integer $x_i \in S_1$ the color $1$ and conversely we
  assign the color $2$ to every vertex corresponding to an integer $x_i \in
  S_2$. In addition we assign $1$ to $A$ and $2$ to $B$.  For every $v_i$ we
  have $d_{c[v_i]}(v_i) = w_i = 2x_i/X - \epsilon/|S| < 1$ as a single integer
  cannot exceed half of the total sum of the integers, otherwise the partition
  would not be valid.  For $A$ we get
  \[
      d_{c[A]}(A) = \sum_{x_i \in S_1} w_i = \frac{2}{X}\sum_{x_i \in S_1}x_i - \frac{\lvert S_1 \rvert}{|S|} \epsilon  = \frac{2}{X}\cdot\frac{X}{2} - \frac{\lvert S_1 \rvert}{|S|} \epsilon = 1 - \frac{\lvert S_1 \rvert}{|S|} \epsilon < 1.
  \]
  The same argument holds for $B$. Hence the partition yields a valid weighted
  improper $2$-coloring of $G$.

  Given a valid weighted improper coloring of $G$, we let $S_1$ be all those
  vertices in the same color class as $A$ and $S_2$ be all those in the same
  color class as $B$. Then
  \[
      \sum_{x_i \in S_1}x_i = \sum_{v_i \in c[A]}\frac{X}{2}\left(w_i + \frac{\epsilon}{|S|}\right) = \frac{X}{2} \left(\sum_{v_i \in c[A]}w_i + \frac{\lvert S_1 \rvert}{|S|} \epsilon\right) \leq \frac{X}{2} \left(1 - \epsilon + \frac{\lvert S_1 \rvert}{|S|} \epsilon\right) \leq \frac{X}{2},
  \]
  and the same holds for $S_2$.
  For both $S_1,S_2$ we get $\sum_{x \in S_i} x \leq \frac{X}{2}$ and by
  definition $\sum_{x \in S_1} x + \sum_{x \in S_2} x = X$, hence equality holds.

  \begin{figure}[h]
    \centering
    \begin{tikzpicture}
       \node[evertex] (1) {\scriptsize$A,v_1,B$};
       \node[evertex] (2) at ($(1) + (2,0)$) {\scriptsize$A,v_2,B$};
       \node[evertex] (n) at ($(2) + (3,0)$) {\scriptsize$A,v_n,B$};
       \node at ($(2) + (1.5,0)$) {$\dots$};
       \draw (1) -- (2);
       \draw (2) -- ($(2) + (1,0)$);
       \draw (n) -- ($(n) - (1,0)$);
    \end{tikzpicture}
    \caption{Path decomposition of the graph constructed from an instance of the partition problem.}
    \label{fig:subsetsumpathdecomp}
  \end{figure}
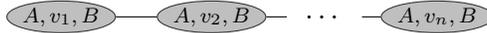

  The graph $G$ has pathwidth $2$ as can be seen by an optimal path
  decomposition of $G$ depicted in Figure~\ref{fig:subsetsumpathdecomp}. Using
  the fact that the partition problem is $\mathcal{NP}$-complete
  \cite{hayes2002computing}, we conclude the proof.
\end{proof}

\section{Bounds}
Prohibiting edges of weight $1$ can make a difference when coloring
weighted graphs. The $3$-regular graph in Figure~\ref{fig:threereg} has
weighted improper chromatic number $3$, even though every vertex has an
incident edge with weight less than $1$.
However, $3$-regular graphs with no edges of weight $1$ have weighted improper
chromatic number at most $2$, as is shown by the following lemma.
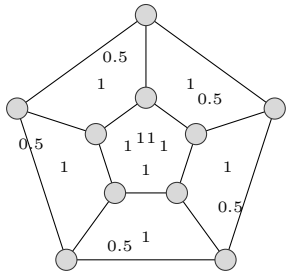
\begin{figure}[h]
  \centering
  \begin{tikzpicture}[auto]
    \foreach \x in {0,...,4}{
      \coordinate (\x) at (\x*72+90:0.7) {};
      \coordinate (0\x) at (\x*72+90:1.8) {};
    }
    \foreach \x in {0,...,4}{
      \draw[edge] (\x) -- node[label={[label distance=-3.0pt]\x*72:\tiny$0.5$}] {} (0\x);
      \pgfmathparse{mod(\x+1,5)}
      \draw[edge] (0\x) -- node[label={[label distance=-3pt]\x*72+126:\tiny$1$}] {} (0\pgfmathresult);
      \draw[edge] (\x) -- node[label={[label distance=-3pt]\x*72+126:\tiny$1$}] {} (\pgfmathresult);
    }
    \foreach \x in {0,...,4}
      \node[wvertex] at (\x) {};
    \foreach \x in {0,...,4}
      \node[wvertex] at (0\x) {};

  \end{tikzpicture}
  \caption{A $3$-regular graph where each vertex has an incident edge with weight less than $1$.}
  \label{fig:threereg}
\end{figure}

\begin{lemma}
  If $G = (V,E,w)$ is a weighted graph such that $w(e) < 1$ for every $e \in E$
  and every vertex has at most $3$ neighbors, then $\chi_w(G) \leq 2$.
\end{lemma}
\begin{proof}
  Start with any $2$-coloring $c:V \rightarrow \{1,2\}$ of $G$. While we have
  some vertex with at least two same-colored neighbors, pick such a vertex $v$
  and flip the color of $v$ to the opposite color. This step reduces the number
  of monochromatic edges by at least $1$ and as there are at most $\lvert E
  \rvert$ monochromatic edges to begin with, the procedure halts after a finite
  number of steps. When the procedure halts, every vertex has at most one
  incident monochromatic edge and as the weights are less than $1$, the degree
  of the vertex with respect to its color is less than $1$.
\end{proof}

We can extend the results of Lov\'asz \cite{lovasz} to our version of the
problem. It is clear that $\frac{\chi(G)}{(d + 1)} \leq \chi_d(G)$
\cite{Kang20081438}, where $\chi_d(G)$ is the $d$-defective chromatic number.
Similar results hold for undirected weighted improper coloring.

\begin{lemma}
  If $G = (V,E,w)$ is an undirected weighted graph with the minimum positive edge weight $w_{min}$, then
  \[
    \ceil{\frac{\chi(G)}{\floor{\frac{1-\epsilon}{w_{\min}}} + 1}} \leq \chi_w(G).
  \]
\end{lemma}
\begin{proof}
  We start by removing all edges $e \in E$ with $w(e) = 0$. Given a weighted
  improper coloring of $G$ with $\chi_w(G)$ colors, each vertex will have at
  most $\floor{\frac{1 - \epsilon}{w_{\min}}}$ neighbors with the same color.
  Hence, by Brooks' theorem, each of the subgraphs induced by the coloring can
  be properly colored with at most $\floor{\frac{1 - \epsilon}{w_{\min}}} + 1$
  colors.  Therefore $\chi(G) \leq \left(\left\lfloor\frac{1 -
  \epsilon}{w_{\min}}\right\rfloor + 1\right)\chi_w(G)$, which yields the
  lemma.
\end{proof}

\begin{theorem}
  \label{prop:dwimp_lovasz}
  If $G = (V,E,w)$ is a weighted digraph, where the underlying graph of $G$ has
  maximum degree $\hat{\Delta}$ and maximum edge weight $w_{max}$, then
  \[
    \chi_w(G) \leq \ceil{\frac{\hat{\Delta}}{\floor{\frac{1 - \epsilon}{w_{\max}}} + 1}} + 1.
  \]
\end{theorem}
\begin{proof}
  We construct a weighted undirected graph $G' = (V', E'=\{\{u,v\}:(u,v)\in
  E\text{ or }(v,u) \in E\},w')$ where $w'(e) = w_{\max}$ for every $e \in E'$.
  Clearly $\chi_w(G) \leq \chi_{w'}(G')$ as we are adding edges and increasing
  the weights. Notice that each $v\in V'$ can have at most $\floor{\frac{1 -
  \epsilon}{w_{\max}}}$ same-colored neighbors in a valid weighted improper coloring. Let $k =
  \ceil{\hat{\Delta}/\floor{\frac{1 - \epsilon}{w_{\max}}+ 1}} + 1$ and start
  with any $k$-coloring of $G'$.

  Say there exists a vertex $v$ with more than $\floor{\frac{1 -
  \epsilon}{w_{\max}}}$ neighbors of the same color. First assume that $v$ has
  more than $\floor{\frac{1 - \epsilon}{w_{\max}}}$ neighbors of every color
  class. But then
  \[
    d_{G'}(v) \geq \left(\floor{\frac{1 - \epsilon}{w_{\max}}}+1\right)\ceil{\frac{\hat{\Delta}}{\floor{\frac{1 - \epsilon}{w_{\max}}} + 1}} \geq \hat{\Delta} + \frac{\hat{\Delta}}{\floor{\frac{1 - \epsilon}{w_{\max}}} + 1},
  \]
  which contradicts the fact that $\hat{\Delta}$ is the maximum degree in the
  underlying graph of $G$. Hence there exists a color class where $v$ has at
  most $\floor{\frac{1 - \epsilon}{w_{\max}}}$ neighbors, and we will change
  the color of $v$ to the color of any such color class.

  We can repeat the previous procedure while there exists a vertex with more
  than $\floor{\frac{1 - \epsilon}{w_{\max}}}$ neighbors of the same color. Now consider
  the monochromatic edges in the graph, which are at most $|E|$ to begin with.
  Every time we change the color of a vertex during this process, the number of
  monochromatic edges decreases by at least one. Hence this process must
  terminate, and at that time each vertex has at most $\floor{\frac{1 -
  \epsilon}{w_{\max}}}$ same-colored neighbors, so we have a valid
  $k$-coloring.
\end{proof}

As the previous proof is constructive, we get the following corollary:

\begin{corollary}
    For a weighted digraph $G=(V,E,w)$ with maximum weight $w_{\max}$ and maximum
    degree $\hat{\Delta}$ in the underlying graph, a weighted improper
    $\left(\ceil{\hat{\Delta}/\floor{\frac{1 - \epsilon}{w_{\max}}+ 1}} + 1\right)$-coloring
    can be found in $O\left(\hat{\Delta}|E|\right)$ time.
\end{corollary}

Halld\'orsson and Bang-Jensen \cite{hallbang} proved the essentially tight
bound of $\chi_w(G) \leq \lfloor2\Delta^{-1}+1\rfloor$ in terms of maximum
weighted indegree. This bound is used in the following theorem to give a bound
in terms of sum of edge weights.

\begin{theorem}
    If $G = (V,E,w)$ is a weighted digraph and $W = \sum_{e \in E}w(e)$, then
    $\chi_w(G) \leq 2\floor{\sqrt{2W}} + 1$.
\end{theorem}
\begin{proof}
  Let $t$ denote the number of vertices $v$ of $G$ with $d^{-}(v) \geq
  \sqrt{W/2}$. Since the sum of indegrees of those vertices is at least $t
  \cdot \sqrt{W/2} \leq W$, we get that $t \leq \sqrt{2W}$. Let $G'$ be
  the induced subgraph of the vertices with indegree at least $\sqrt{W/2}$. By
  \cite{hallbang} the rest of the graph $G \setminus G'$ has chromatic number
  $\chi_w(G \setminus G') \leq \floor{2\sqrt{W/2} + 1} =
  \floor{\sqrt{2W}+1}$. By coloring $G'$ with $\floor{\sqrt{2W}}$
  colors, we get a coloring using at most $\floor{\sqrt{2W}} +
  \floor{\sqrt{2W} + 1}$ colors, and the theorem follows.
\end{proof}

\section{Fixed-Parameter Algorithms}
Both ordinary graph coloring and defective coloring are fixed-parameter
tractable when parameterized by treewidth \cite{kolmanfair,rao2007msol}. Thus it
is natural to ask whether the same holds true for weighted improper coloring.
Theorem~\ref{thm:pathwidthhardness} shows that this is not the case.

\begin{figure}[h]
  \centering
  \begin{subfigure}[b]{0.3\textwidth}
    \begin{tikzpicture}[auto]
      \centering
      \node[wvertex,font=\footnotesize] (v) at (0,0) {$v$};
      \node[wvertex,font=\footnotesize] (u) at ($(a) + (1,1)$) {$u$};
      \node[wvertex,font=\footnotesize] (w) at ($(a) + (1,-1)$) {$w$};
      \node (d) at ($(a) + (-0.7,-0.7)$) {};
      \node (e) at ($(a) + (-0.7,0.7)$) {};
      \node (f) at ($(a) + (-1,0)$) {$\dots$};

      \path[edge,->] (u) --node[yshift=5pt,font=\scriptsize] {$0.9$} (v);
      \path[edge,->] (w) --node[swap,yshift=-5pt,font=\scriptsize] {$0.9$} (v);
      \path[edge]
        (d) -- (v)
        (e) -- (v)
        (f) -- (v);
    \end{tikzpicture}
  \end{subfigure}
  \quad
  \begin{subfigure}[h]{0.3\textwidth}
    \centering
    \begin{tikzpicture}[auto]
      \node[evertex,font=\footnotesize] (a) at (0,0) {$v,u$};
      \node[evertex,font=\footnotesize] (b) at ($(a) + (2,0)$) {$v$};
      \node[evertex,font=\footnotesize] (c) at ($(b) + (2,0)$) {$v,w$};
      \node (0) at ($(a) + (1,0)$) {$\dots$};
      \node (1) at ($(b) + (1,0)$) {$\dots$};
      \path[edge] 
        (a) -- (0)
        (0) -- (b)
        (b) -- (1)
        (1) -- (c);
    \end{tikzpicture}
    \vspace{30pt}
  \end{subfigure}
  \vspace{-28pt}
  \caption{A weighted digraph and its tree decomposition.}
  \label{fig:coloringfptexample}
\end{figure}
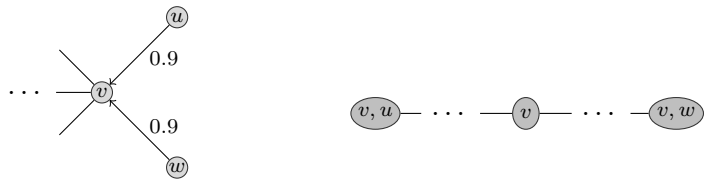

Figure~\ref{fig:coloringfptexample} shows a small part of a graph (on the left)
and what the relation between these vertices might look like in a tree
decomposition of the same graph (on the right). The classic fixed-parameter
algorithm for ordinary coloring would at some point visit the vertex
corresponding to the subset $\{v,u\}$ and choose colors for $v$ and $u$.  Then
it would proceed down the tree decomposition, possibly visiting multiple
vertices, and then finally arrive at the vertex corresponding to the subset
$\{v,w\}$. At that point $w$ needs to be assigned some color, and it is
immediately clear that it cannot be the same color as is assigned to $v$. On
the other hand, if this was weighted improper coloring, then it may or may not
be valid to assign $v$'s assigned color to $w$. This depends on which neighbors
of $v$ have the same color as $v$ does. In this example it depends on the color
of $u$, and possibly other vertices. This shows that valid colors for a vertex
may depend on multiple decisions that are not local to its closest ancestors in
the tree decomposition, possibly giving some intuition into why weighted
improper coloring is not fixed-parameter tractable when parameterized only by
treewidth.

However, we now show that weighted improper coloring is fixed-parameter
tractable when parameterized by treewidth and either the maximum indegree in a
graph or the precision of edge weights. As a first step towards this, we
present the following fact.

\begin{fact}
    \label{lem:treedecompchromatic}
    \textnormal{\cite[Theorem 6]{chlebikova2002partial}} The chromatic number of a graph $G$ of treewidth $k$ is at most $k+1$.
\end{fact}

As ordinary coloring is a stricter requirement than weighted improper coloring,
it follows that the weighted improper chromatic number of a graph $G$ of
treewidth $k$ is also at most $k+1$.

The following algorithms operate on an optimal tree decomposition, which can be
constructed in linear time (when the treewidth is bounded) by Bodlaender \cite{bodlaender1993linear}.

\subsection{Bounded Treewidth and Bounded Maximum Indegree}
If we additionally restrict the maximum unweighted indegree of the graph, then
it is possible to keep track of the colors of all vertices and all their
in-neighbors within a given super-vertex in the tree decomposition, which is
exactly the information needed to check if the coloring restrictions for a
given vertex within that super-vertex are fulfilled.

Now consider any weighted digraph $G=(V,E,w)$ with bounded treewidth $k$
and bounded maximum indegree.
Algorithm~\ref{alg:boundeddegreefpa} is an adaptation of the classic
fixed-parameter algorithm for ordinary graph coloring. It additionally keeps
track of colors of the in-neighbors of the vertices being colored to be able to
check weight restrictions.

A call to \textsc{Color}($X_i$, $\bar{c}$) returns the minimum $l$ such that
the subgraph induced by the vertices contained in the super-vertices in the
subtree rooted at $X_i$ is $l$-colorable with respect to the partial coloring
$\bar{c}$. To answer the query \textsc{Color} tries all possible colorings $c$ of the
vertices and their neighbors using the same colors for those
vertices that have already been colored with the coloring $\bar{c}$. Then it
recursively calls itself to answer further subproblems.

\begin{algorithm}
    \begin{algorithmic}[1]
      \MyFunction{Color}{$X_i$, $\bar{c}$}
        \Begin
            \State Let $X_p$ be the parent of $X_i$ in $T$, or the empty set if $X_i$ is the root
            \State \tbox{Let $V_i$ be the set of vertices in $X_i$ and all their in-neighbors, and let $V_p$ be defined similarly for $X_p$}
            \State $\chi \gets +\infty$
            \For {each coloring $c:V_i \rightarrow \{1,\ldots,k+1\}$ with $c(v) = \bar{c}(v)$ for $v \in V_i \cap V_p$} \label{line:forconsistentcolor}
                \If {$d^{-}_{c[v]}(v) < 1$ for each $v\in X_i$} \label{line:validcoloring}
                    \State $\chi' \gets$ highest index of a color used by $c$
                    \For {each child $X_j$ of $X_i$}
                        \State $\chi' \gets \max(\chi',\, \textsc{Color}(X_j,\, c))$ \label{line:recursivecall}
                    \EndFor
                    \State $\chi \gets \min(\chi,\,\chi')$
                \EndIf
            \EndFor
            \State \Return $\chi$
        \End
        \Statex
        \State Construct a tree decomposition $(X,T)$ of $G$ of width at most $k$
        \State Root the tree $T$ at super-vertex $X_1$
        \State \Return \textsc{Color}($X_1,\, \varnothing$)
    \end{algorithmic}
    \caption{A fixed-parameter algorithm for graphs of bounded treewidth and bounded maximum indegree.}
    \label{alg:boundeddegreefpa}
\end{algorithm}

In order to argue the correctness of this algorithm, there are essentially
three properties that need to be shown:
\begin{itemize}
    \item the assignment of colors to vertices is consistent, that is, exactly
        one color is assigned to a given vertex, and this color is used
        throughout the duration of the algorithm,
    \item any coloring that the algorithm discovers is valid, and
    \item any valid coloring can be discovered by the algorithm.
\end{itemize}

We show that these properties hold in Section~\ref{sec:algboundeddegreecorrect}
of Appendix~\ref{app:correctness}. This gives us the following theorem.

\begin{theorem}
    \label{the:fptdegree}
    Weighted improper coloring is fixed-parameter tractable when parameterized
    by treewidth and maximum indegree.
\end{theorem}
\begin{proof}
    Lemmas \ref{lem:discoveredvalid} and \ref{lem:validdiscovered} show that
    the colorings discovered by Algorithm~\ref{alg:boundeddegreefpa} are
    exactly the valid colorings using at most $k+1$ colors. In particular, as
    the weighted improper chromatic number is at most $k+1$ by
    Lemma~\ref{lem:treedecompchromatic}, it will discover all colorings that
    have the minimum number of colors. And as the algorithm returns the minimum
    number of colors in any coloring found, it will return the weighted
    improper chromatic number of the input graph $G$.

    Algorithm~\ref{alg:boundeddegreefpa} has exponential time complexity, as it
    is essentially a backtracking algorithm that explores all valid colorings
    of at most $k+1$ colors. Now notice that the results of the \textsc{Color}
    function only depend on the values of its input parameters. Also notice
    that the call hierarchy of the recursive calls is acyclic. Therefore the
    results of these function calls can be computed efficiently using dynamic
    programming. That is, by computing the results of the function calls for
    all possible input parameters, storing all results, and then re-using the
    stored results when needed. Let $n=|V(G)|$, $k$ be the treewidth and
    $\hat{\Delta}^{-}$ denote the maximum unweighted indegree of $G$. As there
    are $(k+1)^{(k+1)\hat{\Delta}^{-}}$ possible maps from a set of
    $(k+1)\hat{\Delta}^{-}$ elements to a set of $k+1$ elements, there are
    $O\left(n(k+1)^{(k+1)\hat{\Delta}^{-}}\right)$ possible input parameters to
    the \textsc{Color} function. As the result for each of them can be computed
    in $O\left(n (k+1)^{(k+1)\hat{\Delta}^{-}}\right)$ time, this gives us
    total time complexity $O\left(n^2 (k+1)^{2(k+1)\hat{\Delta}^{-}} \right)$.
    Although unnecessary for this result, noticing that each super-vertex is actually
    only referenced twice reduces the time complexity down to $O\left(n
    (k+1)^{2(k+1)\hat{\Delta}^{-}} \right)$. This time complexity along with
    the correctness of the algorithm proves that weighted improper coloring is
    fixed-parameter tractable when parameterized by treewidth and maximum
    indegree.
\end{proof}

The algorithm has memory complexity $O\left(n(k+1)^{(k+1)\hat{\Delta}^{-}}\right)$.
It may also be of interest that, using common dynamic programming techniques,
it is possible to construct an optimal coloring, or count the number of optimal
colorings.

There are many graphs of bounded treewidth and bounded maximum degree, and this
algorithm is linear when applied to these graphs. This includes bounded degree
interval graphs, $C_n\times K_t$ where $t$ is bounded, and $k$th powers of
paths and cycles where $k$ is bounded. Remember that it is
$\mathcal{NP}$-complete to weighted improper color the unbounded versions of
these graph classes by Lemma~\ref{lem:cliquesnpcomplete}.

\subsection{Bounded Treewidth and Bounded Precision Weights}
Another possibility is to additionally restrict the precision of weights in the
graph. Consider the process of assigning colors to vertices one at a time.
Initially, when no vertex has been colored, all vertices have weighted indegree
zero from same-colored vertices, and are allowed to have an additional weighted
indegree of $1-\epsilon$ of same-colored vertices. We will call this the budget
of the vertex. As more vertices are assigned colors, the budget of a given
vertex gradually decreases. If the weights have bounded precision, then it is
possible to keep track of the budgets in such a process.

Now consider any weighted digraph $G=(V,E,w)$ with bounded treewidth $k$ and
edge weights with bounded precision of $b$ bits. Let $R$ denote the set of
$b$-bit fixed precision real numbers. Algorithm~\ref{alg:boundedweightsfpa},
similar to Algorithm~\ref{alg:boundeddegreefpa}, is an adaptation of the
classic fixed-parameter algorithm for ordinary graph coloring, but instead of
the coloring of the neighbors, it keeps track of vertex budgets to check weight
restrictions.

A call to \textsc{Color}($X_i$, $\bar{c}$, $\bar{r}$) returns the minimum $l$
such that the subgraph induced by the vertices contained in the super-vertices
in the subtree rooted at $X_i$ is $l$-colorable, provided that
\begin{itemize}
\item the partial coloring given by the parameter $\bar{c}$ is the
coloring done at the parent of $X_i$ (which could share some vertices with
$X_i$), and
\item the second parameter $\bar{r}$ gives the remaining budgets of
the vertices in $X_i$'s parent.
\end{itemize}

To answer this query \textsc{Color} tries all possible colorings $c$ of
the vertices in $X_i$ while maintaining the colors given by
$\bar{c}$.  However, it cannot directly call itself recursively to solve
further subproblems. The reason is that the budget of a given vertex is not independent amongst the
children super-vertices. 
In order to break up this dependency, the \textsc{Distribute} function
distributes the budget of a given vertex amongst the children super-vertices.
This is done by recursively going over all child super-vertices of super-vertex
$X_p$ and deciding how much of the budget may be spent in the subtree
rooted at the $i$-th child.
Then \textsc{Color} can be called with this budget, $r$, on the $i$-th
child to solve further subproblems.
%
%

\begin{algorithm}
  \begin{algorithmic}[1]
    \MyFunction{Distribute}{$X_p$, $\bar{c}$, $\bar{r}$, $i$}
      \Begin\,
        \If {$X_p$ has less than $i$ children}
            \State \Return $0$
        \Else
          \State Let $X_i$ be the $i$-th child of $X_p$
          \State $\chi \gets +\infty$
          \For {each possible $r:X_p \rightarrow R$ where $0\leq r(v) \leq \bar{r}(v)$}
            \State $\chi' \gets \max(\textsc{Distribute}(X_p,\,\bar{c},\,(\bar{r} - r),\,i+1),\,\textsc{Color}(X_i,\,\bar{c},\,r))$\label{line:recursivecall2_2}
            \State $\chi \gets \min(\chi, \chi')$
          \EndFor
          \State \Return $\chi$
        \EndIf
      \End
      \Statex
      \MyFunction{Color}{$X_i$, $\bar{c}$, $\bar{r}$}
      \Begin
        \State Let $X_p$ be the parent of $X_i$ in $T$, or the empty set if $X_i$ is the root
        \State $\chi \gets +\infty$
        \For {each coloring $c:V_i \rightarrow \{1,\ldots,k+1\}$ with $c(v) = \bar{c}(v)$ for $v \in V_i \cap V_p$}\label{line:forconsistentcolor2}
          \State Let $r:V_i\rightarrow R$ such that $r(v) = \bar{r}(v)$ if $v\in V_p$ and $r(v) = 1-\epsilon$ otherwise\label{line:assignepsilon}
          \For {each $(u,v) \in E(G)$}
             \If {$\{u,v\} \subseteq X_i$ and ($u\not\in X_p$ or $v\not\in X_p$)}
                \State $r(v) \gets r(v) - w(u,v)$ \label{line:subtractweight}
             \EndIf
          \EndFor
          \If {$r(v) \geq 0$ for each $v\in X_i$} \label{line:checknonnegative} \label{line:validcoloring2}
            \State $\chi \gets \min(\chi, \textsc{Distribute}(X_i,\,c,\,r,\,1))$\label{line:recursivecall2}
          \EndIf
        \EndFor
        \State \Return $\chi$
      \End
      \Statex
      \State Construct a tree decomposition $(X,T)$ of $G$ of width at most $k$
      \State Root the tree $T$ at super-vertex $X_1$
      \State \Return \textsc{Color}($X_1,\,\varnothing,\,\varnothing$)
  \end{algorithmic}
  \caption{A fixed-parameter algorithm for graphs of bounded treewidth and bounded precision weights.}
  \label{alg:boundedweightsfpa}
\end{algorithm}

In order to argue the correctness of this algorithm the same three properties as with Algorithm~\ref{alg:boundeddegreefpa} need to be shown.
We show that these properties hold in Section~\ref{sec:algboundedweightcorrect}
of Appendix~\ref{app:correctness}. This gives us the following theorem.

\begin{theorem}
    \label{the:fptweight}
    Weighted improper coloring is fixed-parameter tractable when parameterized
    by treewidth and precision of weights.
\end{theorem}
\begin{proof}
    Lemmas~\ref{lem:discoveredvalid2} and \ref{lem:validdiscovered2} show that
    the colorings discovered by Algorithm~\ref{alg:boundedweightsfpa} are
    exactly the valid colorings using at most $k+1$ colors. In particular, as
    the weighted improper chromatic number is at most $k+1$ by
    Lemma~\ref{lem:treedecompchromatic}, it will discover all colorings that
    have the minimum number of colors. And as the algorithm returns the minimum
    number of colors in any coloring found, it will return the weighted
    improper chromatic number of the input graph $G$.

    Algorithm~\ref{alg:boundedweightsfpa} has exponential time complexity, as
    it is essentially a backtracking algorithm that explores all valid
    colorings of at most $k+1$ colors. Now notice that the results of the
    \textsc{Distribute} and \textsc{Color} functions only depend on the values
    of their input parameters. Also notice that the call hierarchy of the
    recursive calls is acyclic. Therefore the results of these function calls
    can be computed efficiently using dynamic programming. That is, by
    computing the results of the function calls for all possible input
    parameters, storing all results, and then re-using the stored results when
    needed. Let $n=|V(G)|$, $k$ be the treewidth, and $b$ be the precision of
    weights (i.e.\ the number of bits used to represent them). As there are
    $(k+1)^{k+1}$ possible maps from a set of $k+1$ elements to a set of $k+1$
    elements, there are $O\left(n^2(k+1)^{k+1}2^{b(k+1)}\right)$ possible input
    parameters to the \textsc{Distribute} function. The result for each of them
    can be computed in $O\left(2^{b(k+1)}\right)$ time. Similarly we can see
    that there are $O\left(n(k+1)^{k+1}2^{b(k+1)}\right)$ possible input
    parameters to the \textsc{Color} function and the result for each of them
    can be computed in $O\left(n^2(k+1)^{k+1}\right)$ time. This gives us total
    time complexity $O\left(n^2 (k+1)^{k+1} 2^{2b(k+1)} + n^3 (k+1)^{2(k+1)}
    2^{b(k+1)} \right)$. This time complexity along with the correctness of the
    algorithm proves that weighted improper coloring is fixed-parameter
    tractable when parameterized by treewidth and precision of weights.
\end{proof}

\section{Conclusion}
In this paper we defined the weighted improper coloring problem and presented
some hardness results for it. In particular we showed that weighted improper
coloring is not fixed-parameter tractable when pathwidth is fixed. We
generalized some bounds for defective coloring to weighted improper coloring,
and used a bound by Halldórsson and Bang-Jensen \cite{hallbang} to derive a
bound for weighted improper coloring in terms of the sum of edge weights. As
this bound is not tight, it would be interesting to find a tight bound in terms
of the sum of edge weights.

We also showed that $3$-regular graphs that have edge weights strictly less than $1$
are $2$-colorable, and that two colors may not be sufficient when each vertex
has at most two incident weight-$1$ edges.  It would be interesting to
find necessary and sufficient conditions on edge weights for a $3$-regular
graph to be $2$-colorable. In particular we conjecture that sub-cubic graphs
with at most one incident weight-$1$ edge are $2$-colorable.

We gave fixed-parameter algorithms for weighted
improper coloring when either treewidth and maximum degree are fixed or when
treewidth and precision of edge weights are fixed. These algorithms also imply
a linear-time algorithm for certain graph classes such as interval graphs with
bounded degree.
A combination of rounding edge weights with the fixed-parameter algorithm for
bounded treewidth and bounded precision weights might imply an approximation algorithm
for weighted improper coloring of graphs of bounded treewidth.

This paper presents the results of our final project in our Bachelor studies in
Computer Science and we want to thank our instructor, Magnús Már Halldórsson,
for his guidance, advice and helpful hints. We would also like to thank Henning
Úlfarsson for reviewing.

\bibliographystyle{entcs}
\bibliography{references}

\begin{thebibliography}{10}
\expandafter\ifx\csname url\endcsname\relax
  \def\url#1{\texttt{#1}}\fi
\expandafter\ifx\csname urlprefix\endcsname\relax\def\urlprefix{URL }\fi
\newcommand{\enquote}[1]{``#1''}

\bibitem{andrews1985generalization}
Andrews, J.~A. and M.~S. Jacobson, \emph{On a generalization of chromatic
  number}, Congressus Numerantium \textbf{47} (1985), pp.~33--48.

\bibitem{appel1980every}
Appel, K. and W.~Haken, \emph{Every planar map is four colorable}, Mathematical
  Solitaires \& Games  (1980), p.~145.

\bibitem{araujo}
Araujo, J., J.-C. Bermond, F.~Giroire, F.~Havet, D.~Mazauric and
  R.~Modrzejewski, \emph{Weighted improper colouring}, in: C.~Iliopoulos and
  W.~Smyth, editors, \emph{Combinatorial Algorithms},  Lecture Notes in
  Computer Science  \textbf{7056}, Springer Berlin Heidelberg, 2011 pp. 1--18.

\bibitem{archetti2015a}
Archetti, C., N.~Bianchessi, A.~Hertz, A.~Colombet and F.~Gagnon,
  \emph{Directed weighted improper coloring for cellular channel allocation},
  Discrete Applied Mathematics \textbf{182} (2015), pp.~46--60.

\bibitem{hallbang}
Bang{-}Jensen, J. and M.~M. Halld{\'{o}}rsson, \emph{Vertex coloring
  edge-weighted digraphs}, Inf. Process. Lett. \textbf{115} (2015),
  pp.~791--796.

\bibitem{bermond2010improper}
Bermond, J.-C., F.~Havet, F.~Huc and C.~L. Sales, \emph{Improper coloring of
  weighted grid and hexagonal graphs}, Discrete Mathematics, Algorithms and
  Applications \textbf{2} (2010), pp.~395--411.

\bibitem{bodlaender1993linear}
Bodlaender, H.~L., \emph{A linear time algorithm for finding
  tree-decompositions of small treewidth}, in: \emph{Proceedings of the
  twenty-fifth annual ACM symposium on Theory of computing}, ACM, 1993, pp.
  226--234.

\bibitem{chlebikova2002partial}
Chleb{\'i}kov{\'a}, J., \emph{Partial k-trees with maximum chromatic number},
  Discrete mathematics \textbf{259} (2002), pp.~269--276.

\bibitem{cowen1997defective}
Cowen, L., W.~Goddard and C.~E. Jesurum, \emph{Defective coloring revisited},
  Journal of Graph Theory \textbf{24} (1997), pp.~205--219.

\bibitem{cowencowen}
Cowen, L.~J., R.~H. Cowen and D.~R. Woodall, \emph{Defective colorings of
  graphs in surfaces: Partitions into subgraphs of bounded valency}, Journal of
  Graph Theory \textbf{10} (1986), pp.~187--195.

\bibitem{frick1993survey}
Frick, M., \emph{A survey of (m, k)-colorings}, Annals of Discrete Mathematics
  \textbf{55} (1993), pp.~45--57.

\bibitem{frick1994extremal}
Frick, M. and M.~A. Henning, \emph{Extremal results on defective colorings of
  graphs}, Discrete Mathematics \textbf{126} (1994), pp.~151--158.

\bibitem{harary1985conditional}
Harary, F. and K.~Jones, \emph{Conditional colorability ii: Bipartite
  variations}, Congressus Numerantium \textbf{50} (1985), pp.~205--218.

\bibitem{hayes2002computing}
Hayes, B., \emph{Computing science: The easiest hard problem}, American
  Scientist  (2002), pp.~113--117.

\bibitem{hoefer2014approximation}
Hoefer, M., T.~Kesselheim and B.~V{\"o}cking, \emph{Approximation algorithms
  for secondary spectrum auctions}, ACM Transactions on Internet Technology
  (TOIT) \textbf{14} (2014), p.~16.

\bibitem{timpr-random}
Kang, R.~J. and C.~McDiarmid, \emph{The t-improper chromatic number of random
  graphs}, Combinatorics, Probability and Computing \textbf{19} (2010),
  pp.~87--98.

\bibitem{Kang20081438}
Kang, R.~J., T.~Müller and J.-S. Sereni, \emph{Improper colouring of (random)
  unit disk graphs}, Discrete Mathematics \textbf{308} (2008), pp.~1438 --
  1454, third European Conference on Combinatorics Graph Theory and
  Applications Third European Conference on Combinatorics.

\bibitem{kolmanfair}
Kolman, P., B.~Lidick{\`y} and J.-S. Sereni, \emph{Fair edge deletion problems
  on tree-decomposable graphs and improper colorings}, Manuscript. Available
  from http://orion. math. iastate. edu/lidicky/pub/kls10. pdf  (2010).

\bibitem{lovasz}
Lovász, L., \emph{On decomposition of graphs}, Studia Scientiarum
  Mathematicarum Hungarica \textbf{1} (1966), pp.~237--238.

\bibitem{rao2007msol}
Rao, M., \emph{{MSOL} partitioning problems on graphs of bounded treewidth and
  clique-width}, Theoretical Computer Science \textbf{377} (2007),
  pp.~260--267.

\bibitem{tamura}
Tamura, H., M.~Sengoku, S.~Shinoda and T.~Abe, \emph{Channel assignment problem
  in a cellular mobile system and a new coloring problem of networks}, IEICE
  Transactions on Fundamentals of Electronics, Communications and Computer
  Sciences \textbf{74} (1991), pp.~2983--2989.

\end{thebibliography}

\appendix

\section{Correctness of Algorithms}
\label{app:correctness}

\subsection{Correctness of Algorithm~\ref{alg:boundeddegreefpa}}
\label{sec:algboundeddegreecorrect}

In order to argue the correctness of Algorithm~\ref{alg:boundeddegreefpa},
there are essentially three properties that need to be shown:
\begin{itemize}
    \item the assignment of colors to vertices is consistent, that is, exactly
        one color is assigned to a given vertex, and this color is used
        throughout the duration of the algorithm,
    \item any coloring that the algorithm discovers is valid, and
    \item any valid coloring can be discovered by the algorithm.
\end{itemize}

As a first step towards showing these properties, we will first present some
useful lemmas.

\begin{lemma}
    \label{lem:supervertexbagssubtree}
    If $v\in V(G)$ and $X^v \subseteq X$ is a set of super-vertices such that
    $X_i \in X^v$ if and only if $v \in V_i$, then $X^v$ forms a connected
    subtree of $T$.
\end{lemma}
\begin{proof}
    Take two super-vertices $X_a$ and $X_b$ such that $v\in V_a$ and $v\in
    V_b$. Now we have a few cases:
    \begin{itemize}
        \item If $v\in X_a$ and $v\in X_b$, then by property
            \ref{prop:treedecompsubtree} of Definition~\ref{def:treedecomp},
            $v\in X_c$, and hence $v\in V_c$, for all super-vertices $X_c$ on
            the unique path from $X_a$ to $X_b$.
        \item If $v\in X_a$ and $v\not\in X_b$, then there must exist a vertex
            $w\in X_b$ such that $(v,w)\in E(G)$. By properties
            \ref{prop:treedecompedge} and \ref{prop:treedecompsubtree} of
            Definition~\ref{def:treedecomp} there must exist a super-vertex
            $X_c$ such that $v\in X_d$ for all super-vertices $X_d$ on the
            unique path from $X_a$ to $X_c$, and $w\in X_d$ for all
            super-vertices $X_d$ on the unique path from $X_c$ to $X_b$.
            Therefore $v\in V_d$ for all $X_d$ on the unique path from $X_a$ to
            $X_b$.

            The case where $v\not\in X_a$ and $v\in X_b$ is analogous.

        \item If $v \not\in X_a$ and $v\not\in X_b$, then there must exist a
            vertex $w\in X_a$ such that $(v,w)\in E(G)$ and a vertex $w'\in
            X_b$ such that $(v,w') \in E(G)$. By properties
            \ref{prop:treedecompedge} and \ref{prop:treedecompsubtree} of
            Definition~\ref{def:treedecomp} there must exist two super-vertices
            $X_c$ and $X_d$ (possibly the same) on the unique path from $X_a$
            to $X_b$ such that $w\in X_e$ for all $X_e$ on the unique path
            between $X_a$ and $X_c$ (inclusive), $v \in X_e$ for all $X_e$ on
            the unique path between $X_c$ and $X_d$ (inclusive), and $w'\in
            X_e$ for all $X_e$ on the unique path between $X_d$ and $X_b$.
            Therefore each super-vertex $X_e$ on the unique path between $X_a$
            and $X_b$ either contains $v$, $w$ or $w'$, and so $v\in X_e$.
    \end{itemize}


    In any case $v\in V_c$ for all super-vertices $X_c$ on the unique path from
    $X_a$ to $X_b$, which gives us the lemma.
\end{proof}

\begin{lemma}
    \label{lem:uniquepainter}
    For a given vertex $v\in V(G)$, there is a unique super-vertex $P_v$ that
    decides the color of $v$.
\end{lemma}
\begin{proof}
    Looking at line \ref{line:forconsistentcolor}, we see that the color of
    a vertex $v$ is decided by a super-vertex $X_i$ if and only if $v\in V_i$
    and $v \not\in V_p$. If we consider the subtree $X^v$ given by
    Lemma~\ref{lem:supervertexbagssubtree}, we see that this condition only
    holds for the root of that subtree, and this is the required super-vertex
    $P_v$.
\end{proof}

These two lemmas will be useful when proving the remaining lemmas.

\begin{lemma}
    \label{lem:consistentcolors}
    The assignment of colors to vertices is consistent, that is, exactly one
    color is assigned to a given vertex, and this color is used throughout the
    duration of the algorithm.
\end{lemma}
\begin{proof}
    Take any vertex $v\in V(G)$. Consider the subtree $X^v$ given by
    Lemma~\ref{lem:supervertexbagssubtree} and the super-vertex $P_v$, which is
    also the root of $X^v$, given by Lemma~\ref{lem:uniquepainter}.
    As $v$ appears in $V_i$ for all $X_i\in X^v$, and as color assignments are
    passed down the tree while they are relevant
    (line~\ref{line:recursivecall}) and because no color assignment is
    overwritten (line~\ref{line:forconsistentcolor}), it follows that the color
    assigned to $v$ by the unique $P_v$ is used in every reference to $v$
    throughout the duration of the algorithm.
\end{proof}

\begin{lemma}
    \label{lem:discoveredvalid}
    Any coloring that the algorithm discovers is valid.
\end{lemma}
\begin{proof}
    Lemma~\ref{lem:consistentcolors} shows that the coloring is consistent. Now
    consider a vertex $v\in V(G)$. By property~\ref{prop:treedecompvertex} of
    Definition~\ref{def:treedecomp} there is a super-vertex $X_i$ that contains
    $v$. When the algorithm visits $X_i$, line~\ref{line:validcoloring} assures
    that the weighted indegree of $v$ from same-colored in-neighbors is less
    than $1$. As this holds for every vertex, it follows that any coloring
    found by the algorithm is valid.
\end{proof}

\begin{lemma}
    \label{lem:validdiscovered}
    Any valid coloring of at most $k+1$ colors can be discovered by the
    algorithm.
\end{lemma}
\begin{proof}
    Take any valid coloring $c$ of $G$ of at most $k+1$ colors. Now consider a
    vertex $v$ and the super-vertex $P_v$ given by
    Lemma~\ref{lem:uniquepainter}. When the algorithm visits $P_v$, all valid
    colorings of at most $k+1$ colors of $v$ will be explored (lines
    \ref{line:forconsistentcolor} and \ref{line:validcoloring}). In particular,
    as $c$ is a valid coloring, the coloring where $v$ has color $c(v)$ will be
    explored. As this holds for all vertices $v\in G$, the algorithm will
    discover the coloring $c$.
\end{proof}

These lemmas are enough to prove the correctness of
Algorithm~\ref{alg:boundeddegreefpa}, as is shown by
Theorem~\ref{the:fptdegree}.

\subsection{Correctness of Algorithm~\ref{alg:boundedweightsfpa}}
\label{sec:algboundedweightcorrect}

In order to argue the correctness of Algorithm~\ref{alg:boundedweightsfpa},
there are essentially three properties that need to be shown:
\begin{itemize}
    \item the assignment of colors to vertices is consistent, that is, exactly
        one color is assigned to a given vertex, and this color is used
        throughout the duration of the algorithm,
    \item any coloring that the algorithm discovers is valid, and
    \item any valid coloring can be discovered by the algorithm.
\end{itemize}

\begin{lemma}
    \label{lem:uniquepainter2}
    For a given vertex $v\in V(G)$, there is a unique super-vertex $P_v$ that
    decides the color of $v$.
\end{lemma}
\begin{proof}
    Looking at line \ref{line:forconsistentcolor2}, we
    see that the color of a vertex $v$ is decided by a super-vertex $X_i$ if
    and only if $v\in V_i$ and $v \not\in V_p$. By
    Property~\ref{prop:treedecompsubtree} of Definition~\ref{def:treedecomp},
    the super-vertices in which $v$ appear forms a connected subtree.
    Therefore there is a unique vertex fulfilling this property, the root of
    this subtree, and this is the unique super-vertex $P_v$ we wanted.
\end{proof}

\begin{lemma}
    \label{lem:consistentcolors2}
    The assignment of colors to vertices is consistent, that is, exactly one
    color is assigned to a given vertex, and this color is used throughout the
    duration of the algorithm.
\end{lemma}
\begin{proof}
    Consider any vertex $v\in V(G)$. By Lemma~\ref{lem:uniquepainter2} there is
    a unique super-vertex $P_v$ that decides the color of $v$. By
    Property~\ref{prop:treedecompsubtree} of Definition~\ref{def:treedecomp}
    the set $X^v$ of super-vertices that contain $v$ form a connected subtree
    of $T$. Notice that $P_v$ must be the root of this subtree. As $v$ appears
    in $V_i$ for all $X_i\in X^v$, and as color assignments are passed down the
    tree while they are relevant (lines~\ref{line:recursivecall2} and
    \ref{line:recursivecall2_2}) and because no color assignment is overwritten
    (line~\ref{line:forconsistentcolor2}), it follows that the color assigned
    to $v$ by the unique $P_v$ is used in every reference to $v$ throughout the
    duration of the algorithm.
\end{proof}

\begin{lemma}
    \label{lem:discoveredvalid2}
    Any coloring that the algorithm discovers is valid.
\end{lemma}
\begin{proof}
    Lemma~\ref{lem:consistentcolors2} show that the coloring is consistent. Now
    consider a vertex $v\in V(G)$, and the connected subtree of super-vertices
    that contain $v$ given by Property~\ref{prop:treedecompsubtree} of
    Definition~\ref{def:treedecomp}. When the algorithm visits the super-vertex
    that is the root of this subtree, the vertex $v$ is assigned a budget of $1-\epsilon$
    (line~\ref{line:assignepsilon}). Now consider a vertex $u$ that has an
    outgoing edge to $v$. By the same property as before, the super-vertices
    that contain $u$ form a connected subtree. Furthermore, the set of
    super-vertices that contain both $u$ and $v$ also form a connected subtree.
    Looking at line~\ref{line:subtractweight} we see that, if $u$ and $v$ have
    the same colors, the weight of this edge is subtracted from the budget of
    $v$ at the unique root of the subtree of super-vertices containing both $u$
    and $v$. So the weight of each incoming edge to $v$ of a same-colored
    vertex is subtracted exactly once from $v$'s budget. As
    line~\ref{line:checknonnegative} checks that this quantity never becomes
    negative, it follows that $v$ has weighted indegree from same-colored
    vertices that is less than $1$, so the coloring is valid.
\end{proof}

\begin{lemma}
    \label{lem:validdiscovered2}
    Any valid coloring of at most $k+1$ colors can be discovered by the
    algorithm.
\end{lemma}
\begin{proof}
    Take any valid coloring $c$ of $G$ of at most $k+1$ colors. Now consider a
    vertex $v$ and the super-vertex $P_v$ given by
    Lemma~\ref{lem:uniquepainter2}. When the algorithm visits $P_v$, all valid
    colorings of at most $k+1$ colors of $v$ will be explored (lines
    \ref{line:forconsistentcolor2} and \ref{line:validcoloring2}). In
    particular, as $c$ is a valid coloring, the coloring where $v$ has color
    $c(v)$ will be explored. As this holds for all vertices $v\in G$, the
    algorithm will discover the coloring $c$.
\end{proof}

These lemmas are enough to prove the correctness of
Algorithm~\ref{alg:boundedweightsfpa}, as is shown by
Theorem~\ref{the:fptweight}.

\end{document}